\documentclass[aps,prb,superscriptaddress,11pt]{revtex4-1}
\usepackage{graphicx,amsfonts,amssymb}
\usepackage[outdir=./]{epstopdf}
\usepackage{subcaption} 
\usepackage{grffile} 

\usepackage{listings}
\usepackage{color}
\usepackage{amsmath}
\usepackage{amsthm}
\usepackage{float}

\bibpunct{[}{]}{;}{n}{}{}

\definecolor{dkgreen}{rgb}{0,0.6,0}
\definecolor{gray}{rgb}{0.5,0.5,0.5}
\definecolor{mauve}{rgb}{0.58,0,0.82}

\newtheorem{theorem}{Theorem}

\newtheorem{lemma}[theorem]{Lemma}

\begin{document}
\title{N-break states in a chain of nonlinear oscillators}

\author{A.S.\ Rodrigues}
\affiliation{Departamento de F\'{\i}sica e Astronomia/CFP, Faculdade de Ci\^{e}ncias, Universidade do Porto, R. Campo Alegre,
	687 - 4169-007 Porto, Portugal}
\email{asrodrig@fc.up.pt}

\author{P.G.\ Kevrekidis}
\affiliation{Department of Mathematics and Statistics, University of Massachusetts,
	Amherst MA 01003-4515,  USA}
\email{kevrekid@math.umass.edu}

\author{M.\ Dobson}
\affiliation{Department of Mathematics and Statistics, University of Massachusetts,
	Amherst MA 01003-4515, USA}
\email{dobson@math.umass.edu}

\begin{abstract}
  In the present work we explore a pre-stretched oscillator chain where
  the nodes interact via a pairwise Lennard-Jones potential. In addition
  to a homogeneous solution, we identify solutions with one or more (so-called)
  ``breaks'', i.e., jumps. As a function of the canonical parameter of
  the system, namely the precompression strain $d$, we find that
  the most fundamental one break solution changes stability when the
  monotonicity of the Hamiltonian changes with $d$. We provide a proof
  for this (motivated by numerical computations) observation.
  This critical point separates stable and unstable segments
  of the one break branch of solutions. We find similar branches
  for 2 through 5 break branches of solutions. Each of these
  higher ``excited state'' solutions possesses an additional unstable
  pair of eigenvalues. We thus conjecture that $k$ break solutions
  will possess at least $k-1$ (and at most $k$) pairs of unstable
  eigenvalues. Our stability analysis is corroborated by direct
  numerical computations of the evolutionary dynamics.
\end{abstract}
\date{\today}

\maketitle

\section{Introduction}
\label{sec:intro}

The study of chains with pair-wise interaction potentials
has had a long and distinguished history since the
inception of the Fermi-Pasta-Ulam (FPU) model~\cite{FPU0};
see for some relevant accounts the works of~\cite{FPU1,FPU2}.
Intriguingly, some of the original questions revolving around
the foundations of associated studies remain active topics of investigation
even half a century later. Among them, we note the potential
equipartition of the energy among different degrees of freedom~\cite{flach},
or the number of solitary waves emerging in the
early Kruskal-Zabusky
simulations~\cite{KZ}; for the latter, see the associated
recent work of~\cite{biondini}.

In the present work, we intend to examine a variant of
such inter-site interaction potential chains, in the context
of a  Lennard-Jones (LJ) potential~\cite{LJ}.
We focus, in particular, on the equilibrium states of a pre-stretched, one-dimensional 
LJ chain and provide a detailed
bifurcation analysis of the elastic (i.e., homogeneous) and broken states,
where one or more bonds deviate towards large strains, rendering
the chain inhomogeneous.
We will use the terms broken or fractured to denote the latter bonds.
The Lennard-Jones potential is prototypical 
of non-convex pair interactions, with a convex region for close particles
and a concave region for longer-range interactions, with the force decaying to 
zero as the interparticle distance goes to infinity.  The non-convexity allows
the potential to model fractured states of the material, where two portions
of the chain are sufficiently separated and have very weak interactions,
as is done in $\Gamma$-convergence approaches to the continuum limit 
of such 1D chains~\cite{braid, scha}.

Among the numerous and diverse topics considered for such
LJ lattices, for example the dynamics and mean length
of clusters
at finite temperature~\cite{leedadswell},
the (homoclinic to exponentially small periodic oscillations)
subsonic, as well as (periodic) supersonic lattice traveling
waves~\cite{venney}, the potential for
chaotic motion through the maximum Lyapunov exponent~\cite{okabe},
and as a model for superheated and stretched liquids~\cite{stillinger}
(whereby the role of the
different dynamical configurations must be assessed in the
calculation of thermodynamic quantities).
A linear approximation of the chain and its solutions for nearest-neighbor
(NN) 
and next-near neighbor (NNN) interactions was explored in~\cite{char}.

The existence and stability of one break solutions was studied for the Morse ~\cite{crist84}
and LJ~\cite{stillinger} potentials, while the instability of more than 1 break solutions
was argued. This can be seen intuitively by considering the translation of a
non-boundary segment in a direction that closes one of the fractures.
In both cases the arguments used were based on the relative character
of the energy minimum. In the latter study statistical mechanics arguments were used.
Using the static solutions as initial states, these studies were extended in molecular dynamics methods
to study the expected time for a failure to occur at finite temperatures
(see, e.g.,~\cite{oliveira98} and references therein), 
and collective fluctuations~\cite{lepri05}.
Later, it was shown that for a wide range of potentials and many-neighbor
interactions, the chain
with more than one fracture is always
locally unstable~\cite{ortner} (see also references therein).

Our aim here is somewhat different, as we explore the bifurcation analysis
of different states and provide a systematic count of the eigenvalues
of the different branches of solutions. We also consider the eigendirections
of the relevant instabilities and excite them in order to observe
the dynamical response of the chain to different unstable perturbations
(when appropriate). This helps us shape a systematic picture about the
existence, stability and dynamical properties of the chain.
It adds to the picture provided by molecular dynamics simulations 
by showing  more direct paths to create broken chains,
using eigendirections of the linear excitations.

Our presentation hereafter will be structured as follows. In Section~\ref{sec:setup},
we will present the mathematical formulation and some of the principal
features of the model. In Section~\ref{sec:analysis}, we will prove a basic
result for the stability of the static solutions in connection
to the monotonicity of the Hamiltonian as a function of
the driving precompression parameter. In Section~\ref{sec:numerics_nn}, we will
present numerical computations of existence, stability and dynamics.
Finally, in Section~\ref{sec:conclusion}, we will summarize our findings and present
our conclusions, as well as some directions for future work.

\section{Mathematical Setup: Nearest-neighbor interaction between pre-stretched oscillators chain}
\label{sec:setup}
We consider the following Hamiltonian system describing 
$M$ free oscillators interacting via a potential $\phi(r),$ with the two ends clamped.  
Let $u_n$ for $n = 0, \dots M+1$ denote the displacements of the oscillators, with $u_0=0$ and $u_{M+1}=0.$ 
We also assume that the chain has been pre-stretched to a separation value $d.$ (Bold characters denote
vectors whose components are as in $\mathbf{u}=[u_0, u_1,\dots, u_M,u_{M+1}].$)  The Hamiltonian is 
written as the sum of kinetic and potential energy, giving
\begin{align*}
\mathcal{H}_0(\dot{\mathbf{u}},\mathbf{u}) &= K(\dot{\mathbf{u}}) + V(\mathbf{u})  \qquad \text{where} \\
K(\dot{\mathbf{u}}) &= \sum_{n=1}^{M} \, \frac{1}{2} \dot{u}_n^2,  \\
V(\mathbf{u}) &= \sum_{n=1}^{M+1} \, \left[ \phi(d+u_n-u_{n-1}) - \phi(d)\right]. 
\label{eq:hamiltNN} 
\end{align*}

From this we obtain the equations of motion 
\begin{eqnarray}\label{eq:motion1}
  \ddot{u}_n
           = \phi^{\prime}(d+u_{n+1}-u_n) -\phi^{\prime}(d+u_n-u_{n-1}) \qquad n = 1, \dots, M.
\end{eqnarray}

If we consider the interaction potential to be of the LJ type, scaled to have the dimensionless form:
\begin{eqnarray}
  \phi(r)=\frac{1}{r^{12}} - \frac{2}{r^{6}},
  \label{lj}
\end{eqnarray}
the reference length, where force $f$ is zero, is at $r_0$ such that
\begin{eqnarray}
  \frac{\partial \phi}{\partial r} =f(r)=0
\Rightarrow r_0=1 .
\label{ljeq}
\end{eqnarray}

Similarly, the inflection point is obtained from:
\begin{eqnarray}
  \frac{\partial^2 \phi}{\partial r^2} =0
  \Rightarrow r_i=\left(\frac{13}{7}\right)^{1/6}\approx1.10868.
  \label{infl}
\end{eqnarray}

In our considerations within what follows, we will examine the
possible solutions of the corresponding static problem as parametrized
by $d$. Once a static solution $\mathbf{u}_{0}$ is identified we perturb them
by means of the ansatz:
\begin{eqnarray}
  u_n=u_{0,n} + \epsilon e^{\lambda t} \delta_n.
  \label{pert}
\end{eqnarray}
Substituting in the equation of motion, written as
\begin{equation}
\ddot{u}_n = F_n(\mathbf{u}),\label{eq:mot}
\end{equation}
we obtain
\begin{equation}
\frac{d^2}{dt^2}\left(\mathbf{u_0} + \epsilon e^{\lambda t} \boldsymbol{\delta} \right) =
F(\mathbf{u_0} + \epsilon e^{\lambda t} \boldsymbol{\delta}),
\end{equation}
or
\begin{equation}
\ddot{\mathbf{u}}_0 +\epsilon\lambda^2 e^{\lambda t} \boldsymbol{\delta} = \mathbf{F}(\mathbf{u_0}) + \epsilon e^{\lambda t}
\left. \frac{\partial \mathbf{F}}{\partial \mathbf{u}}\right\|_{u_0} \boldsymbol{\delta} + \mathcal{O}(\epsilon^2)
\end{equation}

At $O(1)$ we obtain 
the steady state equation, and at $O(\epsilon)$ we have:
\begin{equation}
\lambda^2  \boldsymbol{\delta} = \left.  \frac{\partial \mathbf{F}}{\partial \mathbf{u}}\right\|_{\mathbf{u_0}} \boldsymbol{\delta} =
J(\mathbf{u_0}) \boldsymbol{\delta} 
\label{eq:lsa}
\end{equation}
where $J$ is the Jacobian matrix. This is an 
eigenvalue problem arising
for the eigenvalue-eigenvector pair $(\lambda,\boldsymbol{\delta})$.
The relevant i-th pair $(i=1, \dots, M)$
will also be denoted by $(\lambda_i,\mathbf{e}_i)$
in what follows.
Its
result
will allow us to assess the spectral (linear) stability
of the different solutions, as a non-vanishing real part
of the eigenvalue $\lambda$ (positive $\lambda^2$) will be associated with dynamical
instability (the perturbation will grow), while for marginally stable solutions (the perturbation will just oscillate) all $\lambda$'s will
lie on the imaginary axis (negative $\lambda^2$).

\section{Bifurcation analysis of the Lennard-Jones chain: A Criterion}
\label{sec:analysis}

  Before we embark on a systematic numerical computation of the stationary
  solutions and their spectral properties, we establish a theoretical
  criterion for stability motivated by our numerical computations that
  will follow.
Due to the nearest-neighbor interactions between the particles, the equilibrium
states are particularly simple, as the balance of force on each particle gives
\begin{eqnarray}
  \phi'(d + u_{n+1} - u_n)  = \phi'(d + u_n - u_{n-1}).
  \label{equil}
  \end{eqnarray}
We define the bond length (or strain) variables $r_n = d + u_{n+1} - u_n,$ where
we have the equilibrium condition 
$$
\phi'(r_{n})  = \phi'(r_{n-1}), \qquad n = 1, \dots, M.$$
The Dirichlet boundary conditions $u_0 = u_{M+1} = 0$
correspond to the total strain condition 
\begin{eqnarray}
\sum_{n = 0}^M r_n = (M+1)d.
\label{equil2}
\end{eqnarray}

 We let $f_{\rm max} = \max_r \phi'(r) = \phi'(r_i).$
 For $0 < f < f_{\rm max},$ there
are two solutions to $\phi'(r) = f,$ one with 
$1< r < r_i$, the other with
$r > r_i.$  We describe a bond with length $r < r_i$ as elastic and with 
length $r > r_i$ as broken, and define the two right inverses
$r_e : (0, f_{\rm max}] \rightarrow (1, r_i]$ and $r_b : (0, f_{\rm max}] \rightarrow [r_i, \infty)$
where $\phi'(r_e(f)) = \phi'(r_b(f)) = f$ for every $f \in (0, f_{\rm max}].$
In the following, we will consider equilibria containing one or more 
breaks.

\begin{lemma}
There is a minimal $d$ for which one break solutions exist, corresponding to $M+1$ isolated saddle points 
at which the stability of the equilibria changes.
\end{lemma}

\begin{proof}
From the elastic and broken bond solutions, we can parameterize the 
equilibrium states using the bond stress $f.$  
Then, a uniformly stretched chain has
total length $L_0(f) = (M+1) r_e(f)$ whereas a chain with a single break
has total length $L_1(f) = M r_e(f) + r_b(f).$  Note that
$L_0(f)$ is a monotone increasing function of $f$, whereas $L_1$ is not,
it has a local minimum when $L_1'(f) = M r_e'(f) + r_b'(f) = 0.$  

The total energy for a chain with a single break is
$\mathcal{H}_1(f) = M \phi(r_e(f)) + \phi(r_b(f)),$ and
we see   
\begin{eqnarray}
\mathcal{H}_1'(f) = M \phi'(r_e(f)) r_e'(f) + \phi'(r_b(f)) r_b'(f) =
M f \, r_e'(f) + f \, r_b'(f),
\label{equil3}
\end{eqnarray}
so that its local minimum corresponds to that of $L_1.$  We can also show directly that this point
represents a change in stability for the single fracture solution.

For that, consider a single-fracture equilibrium with strain $\mathbf{r},$ where we take without
loss of generality $r_0 = r_b$ and $r_n = r_e$ for $n = 1, \dots, M.$  We will denote the 
internal stress $f = \phi'(r_e) = \phi'(r_b).$  We then consider
the mean-zero perturbation direction 
$$\delta_n = \begin{cases}
M & n = 0, \\
-1 & n = 1, \dots, M.
\end{cases}
$$
When applying a perturbation $\epsilon \delta_n,$ for positive $\epsilon,$ this enlarges the break, proportionally
shrinking the rest of the chain, and inversely for negative $\epsilon.$  For large enough $d,$ there are two single-fracture equilibria possible, one stable and one unstable, with the unstable branch moving toward the no break solution
for negative epsilon and toward the stable single-fracture equilibrium for positive epsilon, see Figure~\ref{fig:cont1} below.
Then we expand the energy
\begin{align*}
H(\mathbf{r} + \epsilon \mathbf{\delta}) &= 
\phi(r_b + M \epsilon) + M \phi(r_e - \epsilon) \\
&= \phi(r_b) + M \phi(r_e) + 
[M \phi'(r_b) - M \phi'(r_e) ] \epsilon
+ [M^2 \phi''(r_b) + M \phi''(r_e)] \frac{\epsilon^2}{2} + O(\epsilon^3)
\end{align*}
The linear term in $\epsilon$ is zero since $\mathbf{r}$ is an equilibrium.
The quadratic coefficient satisfies
\begin{align*}
[M^2 \phi''(r_b) + M \phi''(r_e)] &= M \phi''(r_b) \phi''(r_e) 
\left[ \frac{M}{\phi''(r_e)} + \frac{1}{\phi''(r_b)} \right] \\
&= M \phi''(r_b) \phi''(r_e) 
\left[ M r'_e(f)  + r'_b(f) \right].
\end{align*}
where the last equality follows from differentiating $\phi'(r_e(f)) = \phi'(r_b(f)) = f.$
This has a zero (that is, a change of concavity) exactly when $L'_1(f)$ does as well.
\end{proof}

This calculation suggests that motion along this eigendirection becomes
neutral as we cross the relevant critical point of the length or energy
curve as a function of the precompression parameter $d$. As a result,
crossing this point will induce a change of stability along the corresponding
eigendirection, a feature that we will monitor further in our detailed
computations below.
It is relevant to point out here that the stability criterion
developed herein
is in line with recent criteria (based on energy monotonicity
changes upon suitable parametric variations) for stability of both traveling
waves in lattices~\cite{hai1,hai2} and breather-like periodic
orbits~\cite{dep}.

Note that the criterion proved above is applicable to any potential that has a change of
concavity and a maximum for the absolute value of the force.

\section{Numerical results for the Nearest-Neighbor Lennard-Jones Potential}
\label{sec:numerics_nn}

\subsection{Steady state}

In our existence computations, we identified stationary solutions via
a fixed point (Newton) iteration scheme. Using Eq.~(\ref{eq:lsa}), we also calculate the
eigenmodes 
$\mathbf{e}_i$
and corresponding eigenvalues $\lambda_i$
of that configuration (where the index in both $\lambda$ and $\mathbf{e}$ labels an ordering,
which we choose to be of decreasing magnitude of $\lambda^2$).
Upon identifying a member of
a particular family of solutions (with one or more breaks/fractures),
we performed a continuation in the parameter $d$. When a turning point was reached, 
the direction of variation of $d$ was reversed, and care was taken to ensure the segment of the
curve followed was a different one (see Figs.~(\ref{fig:cont1}--\ref{fig:cont3a})).
A more detailed description of the numerical procedure can be found, e.g., in Ref.~\cite{numer}.

In what follows we will be showing results obtained for a chain with 20 free nodes.

\subsubsection{1 break}

\begin{figure}[H]
\begin{center}
\includegraphics[width=12cm]{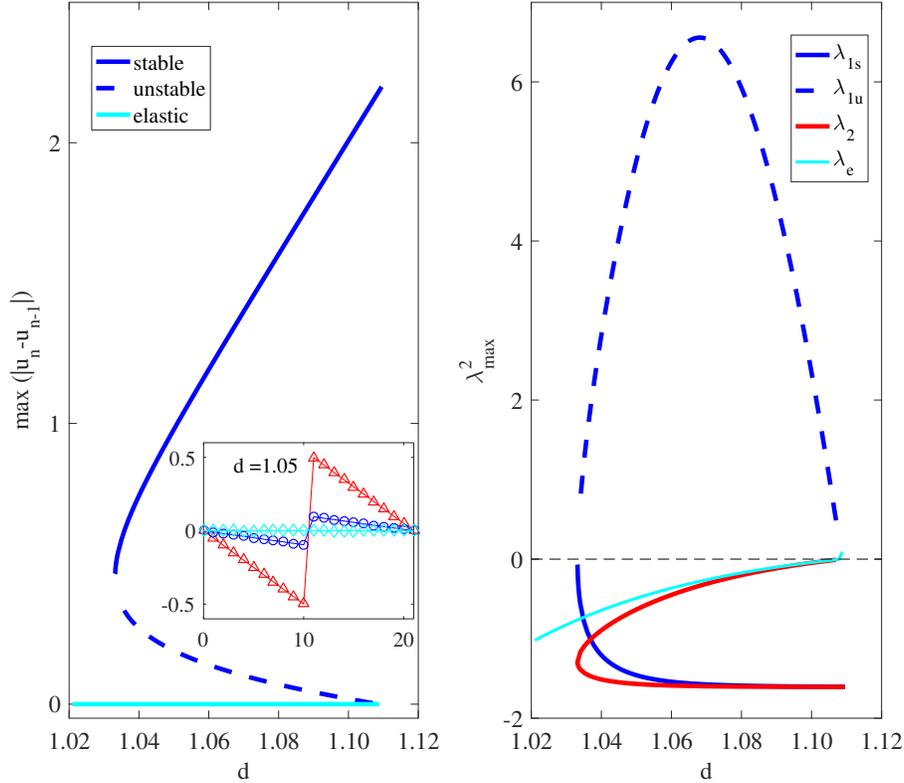}
\caption{The left panel shows stable (blue, solid) and unstable (blue, dashed)
modes with 1 break as a function of the pre-stretching parameter $d$. Also
shown (cyan) is the uniform stretch mode (no breaks). The inset represents
examples of the profiles with an elastic (cyan), an unstable break (blue), and a 
stable break (red).  As $d$ grows, the unstable mode merges with the
uniform stretch mode. The right panel shows the largest two eigenvalues of the
1 break (again, solid for the stable part, dashed for the unstable), together
with the largest eigenvalue of the elastic (uniform stretch) mode. The numbers
in subscripts indicate the order of the eigenvalue, while the subscript
letters indicate
stability with s standing for stable or and u for unstable.
}
\label{fig:cont1}
\end{center}
\end{figure}

In Fig.~\ref{fig:cont1}, left panel, we have represented the amplitude of the broken bond as
a function of the continuation parameter $d$. Two modes were found, a stable (blue, solid), and an unstable
one (blue, dashed). The elastic  (no breaks) mode is also shown (cyan,solid). The inset shows the corresponding profiles for
select values of $d$. These broken
states only exist above $d=1.034$, the turning point of
the branch.
The unstable one break branch can also be identified in the graph
as bifurcating from the uniform elastic solution
at the critical strain $d = r_i.$

In the right panel we show the highest eigenvalue for each $d$, and also the second highest if the  mode is unstable. As per the analysis
of the previous section, 
the monotonicity change of the maximal strain is correlated
with the stability change
of the one break solutions. We have indeed confirmed that the maximal strain,
as well as the total length of the chain but also, importantly, the
energy of the configuration all have turning points at the location
of the change of stability of the branch. In particular, the monotonically
increasing portion of the branch is associated with stability, while
the monotonically decreasing one with instability.
Let us now see how the situation is modified
in the presence of an additional break.

\subsubsection{2 breaks}

\begin{figure}[H]
\begin{center}
\includegraphics[width=12cm]{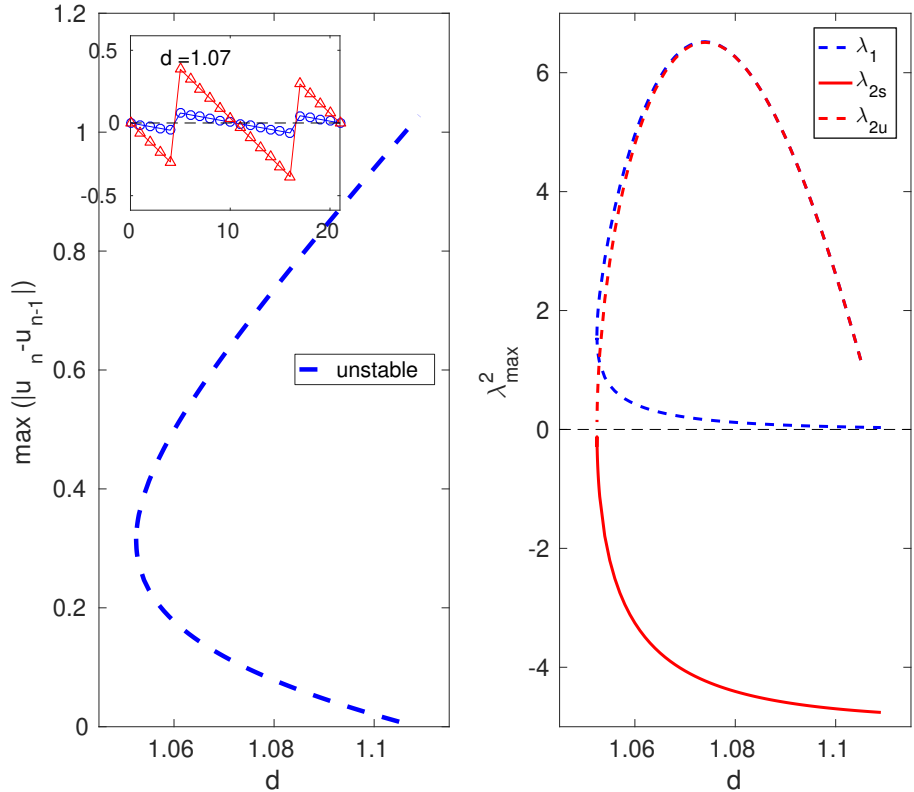}
\caption{The left panel shows the amplitude of the 2 break solutions
  as a function of pre-stretching, $d$. The inset represents
  an example of the two profiles for a given $d$. The right panel shows the two largest eigenvalues
  of this branch (both for
  its energy/strain increasing and decreasing segments),
  which is always unstable. Note that it only exists for a higher
  pre-stretching than the 1 break modes. Subscripts in legend as in Fig.~\ref{fig:cont1}}
\label{fig:cont2c}
\end{center}
\end{figure}

The configuration with 2 breaks was found to be always unstable;
see Fig.~\ref{fig:cont2c}.
In this case too, the branch was found to possess two segments,
one of which with two unstable eigenvalue pairs (the additional
one stemming from the monotonicity of the energy as a function of
the precompression strain $d$), while the other one with only
one unstable eigenvalue pair. 
These two branch segments once again terminated
in a saddle-center bifurcation at a critical value of $d$, higher than
that of the 1 break branch.

\subsubsection{3 breaks}
\begin{figure}[H]
\begin{center}
\includegraphics[width=12cm]{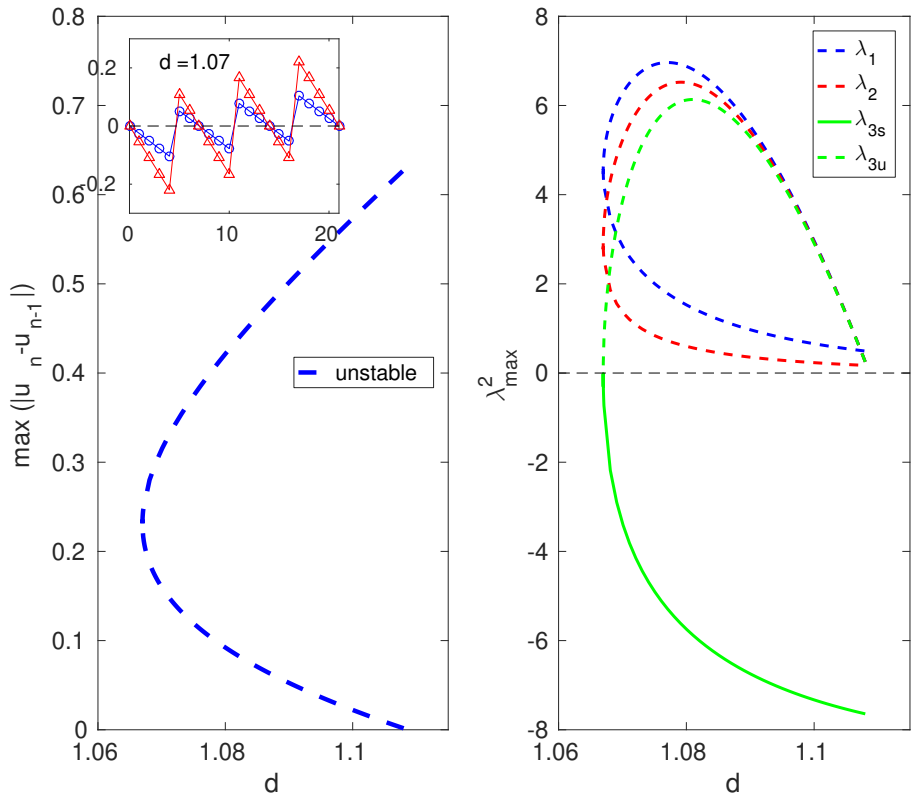}
\caption{The left panel shows the amplitude of the 3 breaks solutions
  as a function of pre-stretching, $d$. The inset represents an example of
  the solution profiles for a given $d$. The right panel shows the three
  largest eigenvalues associated with the saddle- and
  center- portions of the 3 break branch, which are both always unstable. Subscripts in legend as in Fig.~\ref{fig:cont1}}
\label{fig:cont3a}
\end{center}
\end{figure}

Similar conclusions could be drawn for the case with $3$ breaks;
see Fig.~\ref{fig:cont3a}.
Here, the different segments of the branch generically possessed two
unstable eigenvalue pairs. The one with the monotonically increasing
dependence on $d$ had only these two unstable modes, while the decreasing one,
just as before, featured an additional pair of unstable eigenvalues.
From this, as well as our additional results involving modes up
to $N=5$ breaks, a general picture is emerging regarding the stability
properties of the different branches. As illustrated in section III,
the change
of monotonicity of the energy is associated with a change of stability of 
a particular eigenmode.
For the relevant eigenmode,
an increasing energy as a function of $d$ results in stability (along
this eigendirection) while a decreasing energy leads to instability.
In addition to this eigendirection,
the presence of $N$ breaks
implies the existence of an additional $N-1$ unstable eigenmodes.
These features are summarized in Fig.~\ref{fig:combined}
showcasing the dependence of the maximal strain as well as
of the energy on the precompression strain $d$. Now, we discuss
the implications of the excitation of the corresponding unstable
eigenvectors, as a preamble towards predicting the dynamical
evolution of the associated instabilities.

\begin{figure}
	\centering
	\includegraphics[width=12cm]{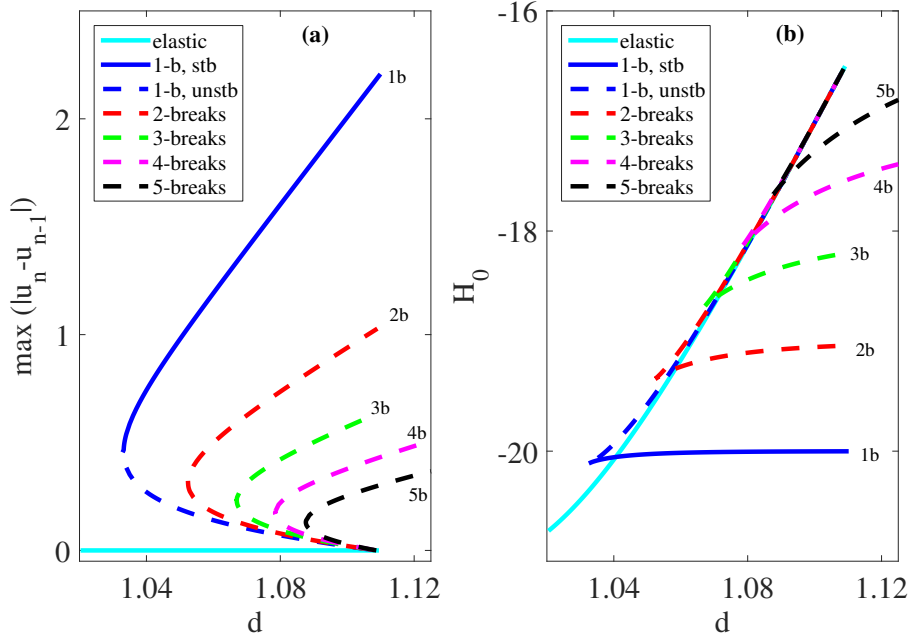}
	
	\caption{Combined results for elastic and 1-5 break branches.
		The left panel
		summarizes our results for break length, while the
		right panel shows the energy of these branches as a function
		of the potential pre-stretching parameter $d$.
	}
	\label{fig:combined}
\end{figure}

\subsubsection{Geometry of the Principal Eigenmodes}

Our aim in the present section is to explore the exact stationary solutions
when the unstable eigenmodes are appended to them, in order to appreciate
the paths that the system can take towards
the decay of the unstable stationary states.
The next series of plots show the modes found above together with the eigenvectors
that are associated with their potential instabilities, as identified before.
The right panels show a linear combination of the mode with a small perturbation in the form of each eigenvector represented on the left panels. In
Fig.~\ref{fig:u_ev1s} and the following similar figures, the weight given to the perturbation was
exaggerated for clarity. On the corresponding dynamical simulations,
small weights were applied,
consistent with the linear stability hypothesis behind the
calculation leading to those eigenvectors.
This is shown for the upper and lower segments of
branches in the top and bottom rows,
respectively.

\begin{figure}[H]
	\begin{center}
		\includegraphics[width=12cm]{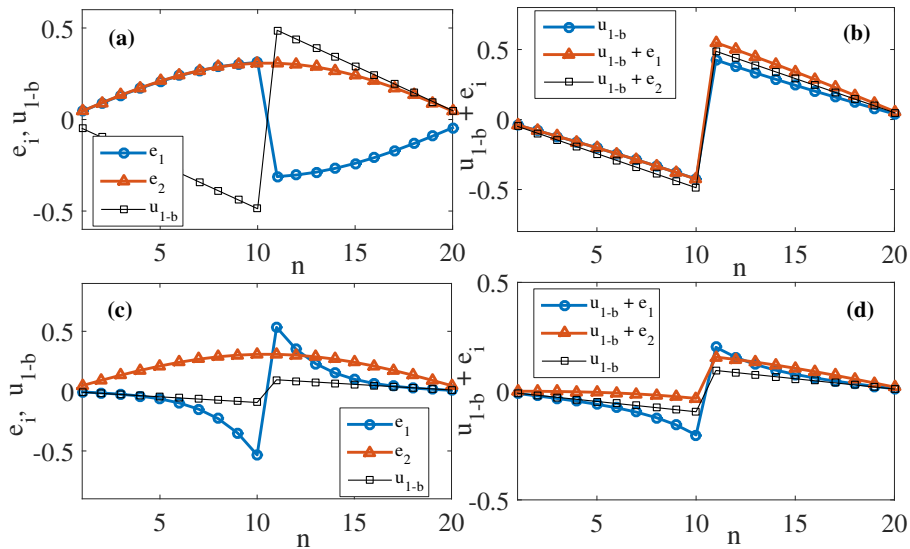}
		\caption{For the 1 break branch, the
                  upper (top row) and lower (bottom row) segment eigenvectors
                  ($\hat{\mathbf{e}}_{1,2}$)
                  are shown in the left panel (the two principal
                  ones). The
                  right panels show the
                  1 break solutions with $0.2\times\hat{\mathbf{e}}_{1,2}.$
              Here $d=1.05$, as in the inset of Fig.~\ref{fig:cont1}}. 
		\label{fig:u_ev1s}
	\end{center}
\end{figure}

In Fig.~\ref{fig:u_ev1s} we show the effect of the eigenvector corresponding to the largest eigenvalue on the shape of the modes, for upper (linearly
stable) and lower (unstable) single break branch segments.  In this case we also show the second eigenvector for illustration, but it always has $\lambda^2<0$, so its effect will be
oscillatory (i.e., the mode will be marginally stable and will not lead
to instability). At first sight the eigenvectors seem to have opposite effects, but we can always perform a phase shift of $\pi$ (since there is the freedom
of multiplying them by any real constant). The important difference lies on the sign of $\lambda^2$, which is negative for the upper branch segment, and so
its effect is to solely lead to a benign oscillation, while for the lower
branch segment it grows with time. It is this growth that leads to destruction of the mode.
The decay can lead to 2 distinct results, as will be shown below:
in the form shown, $\hat{\mathbf{e}}_1$ will make the unstable
state $u_{1l}$ (subscript $l$ for the lower segment branch and $u$ for the upper branch segment) grow towards a stable 1 break
waveform on the upper branch segment, albeit a oscillating one, given the non-dissipative, Hamiltonian
nature of the model. However, if we change
the sign of the perturbation it will decay to the
elastic mode, shedding some energy
in the form of small amplitude waves in the process.

\begin{figure}[H]
	\begin{center}
		\includegraphics[width=12cm]{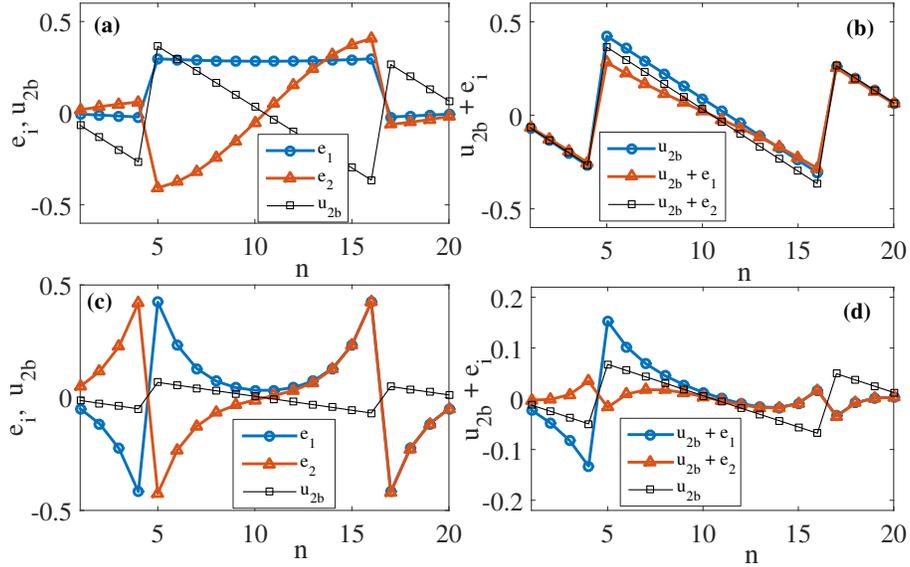}
		\caption{Similar to the previous figure, but now
                  for the 2 break branch, now for $d=1.07$, as in Fig.~\ref{fig:cont2c}. Again, in the right panel, a perturbation involving the
                  relevant eigenvectors,
                  $0.2\times\hat{\mathbf{e}}_{1,2}$, has been
                  added to the two segments (increasing/decreasing in
                top/bottom, respectively) of the branch.} 
		\label{fig:u_ev2s}
	\end{center}
\end{figure}

\begin{figure}[H]
	\begin{center}
		\includegraphics[width=14cm]{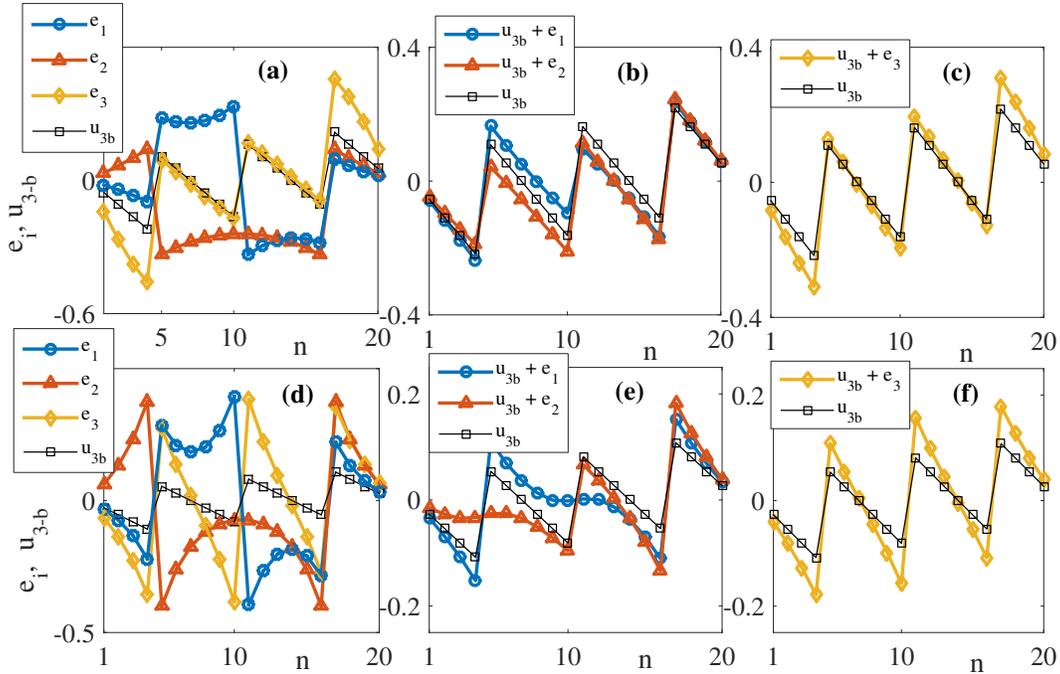}
		\caption{Similar to the previous figures, but now
		  for the 3 break branch, and $d=1.07$, corresponding to the inset of Fig.~\ref{fig:cont3a}. Here too, a perturbation
                  in the direction of the leading eigenvectors
                  of the form $0.2\times\hat{\mathbf{e}}_{1,2,3}$ was added to the upper (top row) and lower (bottom row) segments of the branch (middle and right columns). The left column shows the mode profile
                  and the first 3 eigenvectors; the middle column shows
                  the linear combination of the mode with the first or second eigenvector; the right column shows the linear combination of the
                  mode with the third eigenvector. This last one has different signs in upper or lower segments.}
		\label{fig:u_ev3s}
	\end{center}
\end{figure}

In Fig.~\ref{fig:u_ev2s}
we show a similar representation for the 2 break case.
The two leading eigenvectors alternate in parity with respect to the breaks, 
and so it is expected that  they
appear to seed different dynamical evolutions.  For example, for the
unstable lower branch, one of
these eigendirections involves the two breaks moving in concert, either moving
towards
the larger 2 break solution or the uniform state.

On the other hand,
addition of the other eigendirection (the one that is generically unstable)
tends to convert the 2 break state into a 1 break one, i.e., to eliminate
one of the two breaks. Similar interpretations
can be generalized in the case of
the 3 break solution, with the only difference that now there are two
generically unstable eigendirections, tending to reduce the number
of breaks in the system.

The analogous representation for the 3 break mode is shown in Fig.~\ref{fig:u_ev3s}.
Here, the most unstable eigenmode $\hat{\mathbf{e}}_1$ for the upper segment is either ``in-phase" (IP)
with the side breaks and ``out-of-phase" (OOP) with the central one (as represented in left upper panel
of the Fig.~\ref{fig:u_ev3s}, blue circles)
or vice-versa.
This causes the elimination of the central break,
allowing for the survival of the lateral
ones, if added (as represented in the middle top panel of the same Fig.~\ref{fig:u_ev3s}), or induces the decay the lateral ones,
and survival of the one in the middle, that grows to a stable oscillating 1 break, if subtracted.
The second most unstable eigenmode $\hat{\mathbf{e}}_2$ has a different parity (see again left upper
panel, but now the red triangles), so it is natural
to expect that whether added or subtracted will essentially
lead to a qualitatively similar result. Again from the middle top panel, we see that it will initially
reduce one of the lateral breaks, and increase the size of the other, the middle one remaining
essentially unchanged. 
As for the third eigenmode, it is stable for this upper segment, i.e., will not lead to growth or decay,
but only oscillation. From the right upper panel we see that its
effect is more pronounced on the
lateral breaks.

For the lower segment, from the lower left panel we can see that the general characteristics of the 3
eigenmodes represented do not differ from those of the upper segment. Given the smaller size of the mode
of the lower segment, however, its effects can be more pronounced. This is apparent on the middle and
right lower panels, where the central (for $\hat{\mathbf{e}}_1$), or left (for $\hat{\mathbf{e}}_2$)
breaks have essentially disappeared.
Now the third eigenvalue is  also unstable. So
the effects of the highest two eigenmodes should be qualitatively the same as for the upper segment.
The third eigenvalue however, can show changes, as now it can lead to decay of all 3 breaks (if we have
it OOP with the mode, i.e., oppositely to the situation represented). 

We now turn to the dynamical evolution of the branches,
armed with the interpretation of the different unstable states
and their associated eigendirections.

\subsection{Dynamics}

\begin{figure}[H]
	\begin{center}
		\includegraphics[width=0.95\textwidth]{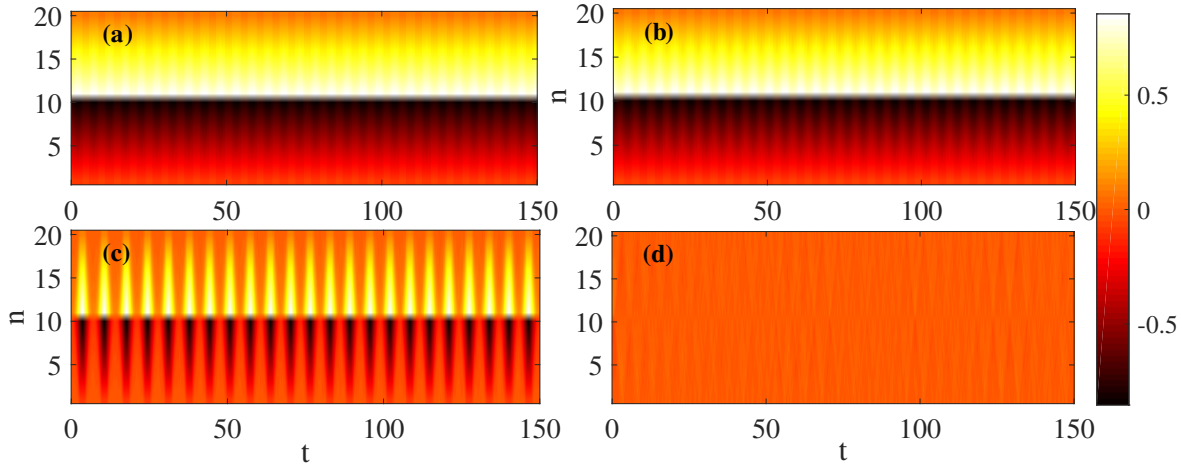}
		\caption{The dynamical evolution of
                  the 1 break branch is shown in spatio-temporal
                  ($n-t$) contour-plot form of the displacements.
                  The initial condition consists of the
                  stationary solution with a perturbation of
                  $\pm 0.01\times\hat{\mathbf{e}}_1$ added to it,
                  for $d=1.05$. (a) upper (linearly stable) segment branch mode $u_{1u} + 0.01\times\hat{\mathbf{e}}_1$; (b)  upper segment branch mode $u_{1u}- 0.01\times\hat{\mathbf{e}}_1$; (c) lower (unstable)  segment branch mode $u_{1l}+ 0.001\times\hat{\mathbf{e}}_1$; (d) lower segment branch mode $u_{1l} - 0.001\times\hat{\mathbf{e}}_1$.
                  In the latter two, the instability leads, respectively,
                  to oscillations around
                  the upper segment branch and to
                  degeneration to the homogeneous state.}
		\label{fig:dyn1s}
	\end{center}
\end{figure}

We start by illustrating the potential outcomes of the evolution
of a 1 break state.
In Fig.~\ref{fig:dyn1s} we show the evolution of such a state at the value corresponding to
the profiles shown in Fig.~\ref{fig:u_ev1s} i.e., for $d=1.05$.
On the upper row we start with an upper branch segment (stable) 1 break mode. We can see that, even with 
 a moderate perturbation (in this case a component proportional to the eigenvector of the largest
 eigenvalue), the waveform is
able to maintain its shape for the duration of the propagation, although there is some oscillation
due to the extra energy stemming from the perturbation. We
ensured that the numerical scheme conserved the initial
energy throughout the propagation duration.

On the other end, the bottom panels show the evolution starting with the
unstable 1 break solution for the same $d$. Here
the amount of perturbation introduced was much smaller (by an order of magnitude), and yet very quickly this 1 break decays.
Importantly, however, the two distinct evolutions of panels (c) and (d)
illustrate that depending on the direction of the perturbation, the
unstable 1 break (operating as a separatrix) may lead
either towards the stable 1 break
branch segment (featuring large amplitude oscillations)
or towards a homogeneous state.
These two radically different behaviors shown in panels (c) and (d)
confirm what was hinted on Fig.~\ref{fig:u_ev1s}: adding the most
unstable eigenvector takes the system to the stable 1 break solution, while
subtracting takes it to the elastic state.
It is interesting to point out that even without introducing any noise
explicitly, the numerical round-off error would eventually lead the
configuration to decay.

\begin{figure}[H]
	\begin{center}
		\includegraphics[width=0.95\textwidth]{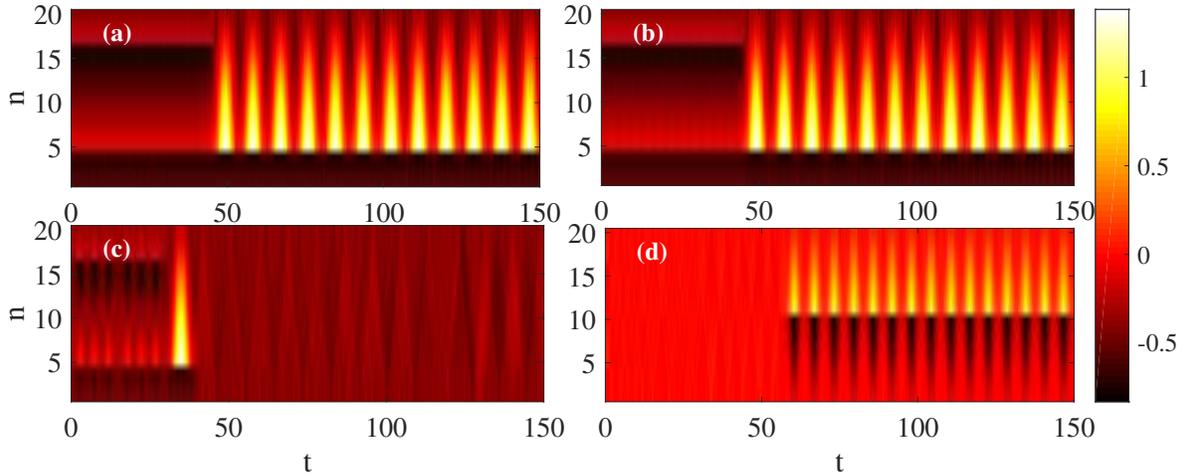}
		\caption{ Similar to Fig.~\ref{fig:dyn1s}.
              Here the initial condition consists of the
                  2 break waveforms with a perturbation added in the form
                  of the second eigenmode, $\pm 0.01\times\hat{\mathbf{e}}_2$,           
                  for $d=1.07$. (a) Upper segment branch mode $u_{2u} + 0.01\times\hat{\mathbf{e}}_2$;
                  (b)  upper segment branch mode $u_{2u}- 0.01\times\hat{\mathbf{e}}_2$; (c) lower
                  segment branch mode $u_{2l}+ 0.01\times\hat{\mathbf{e}}_2$; (d) lower segment branch
                  mode $u_{2l} - 0.01\times\hat{\mathbf{e}}_2$. Notice that although we
                  perturb the wave in the direction of the less unstable
                  eigenmode $\hat{\mathbf{e}}_2$, the more unstable
                  one ($\hat{\mathbf{e}}_1$) eventually
                  crucially contributes to the
                  destabilization dynamics of
                  both segments of the 2 break branch.}
		\label{fig:dyn2s}
	\end{center}
\end{figure}

In Fig.~\ref{fig:dyn2s} we represent  now the result of propagation of a perturbed 2 break solution,
corresponding to the profiles shown in Fig.~\ref{fig:u_ev2s}, for which $d=1.07$. The main difference now
is that the highest eigenvalue is positive for both branch segments, and so it dominates the motion.
As a result, although we perturb only with the second eigenvector
(which is only unstable for the lower segment of the branch),
even the upper branch segment
suffers decay, because of numerical noise, although it takes longer to develop. Thus adding or
subtracting the second eigenvector leads essentially to a (later) decay into a 1 break. For the lower
branch adding $\hat{\mathbf{e}}_2$ should
lead to an oscillation around an upper branch 2 break waveform, yet
the effect of contamination by a $\hat{\mathbf{e}}_1$ causes one of them to
decay. Subtracting $\hat{\mathbf{e}}_2$ should lead to an elastic mode, and that's what the simulation
shows during an initial stage. However the energy present is enough to eventually ``nucleate''
a stable 1 break.

\begin{figure}[H]
	\begin{center}
		\includegraphics[width=0.95\textwidth]{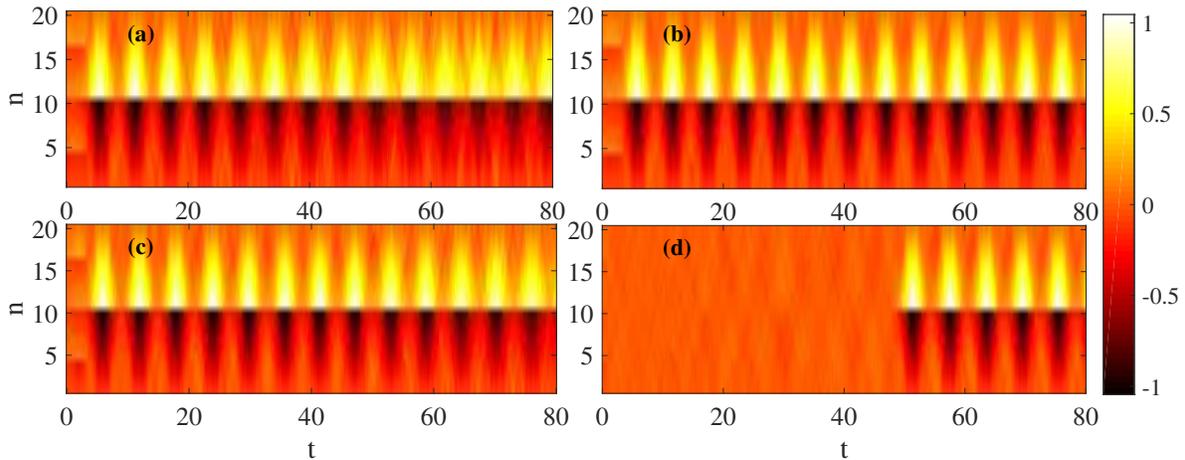}
		\caption{Similar to the previous figures,
                  but now for a 3 break branch with a
                  perturbation $\pm 0.05\times\hat{\mathbf{e}}_3$, for $d=1.07$. (a) Upper segment of the
                  branch mode $u_{3u} + 0.05\times\hat{\mathbf{e}}_3$; (b)  upper segment of the branch
                  mode $u_{3u}- 0.05\times\hat{\mathbf{e}}_3$; (c) lower segment of the branch mode
                  $u_{3l}+ 0.05\times\hat{\mathbf{e}}_3$; (d) lower segment of the branch mode
                  $u_{3l} - 0.05\times\hat{\mathbf{e}}_3$.
                  In all four cases, eventually the dynamics results in a
                1 break state.}
		\label{fig:dyn3s}
	\end{center}
\end{figure}
\begin{figure}[H]
	\begin{center}
		\includegraphics[width=0.95\textwidth]{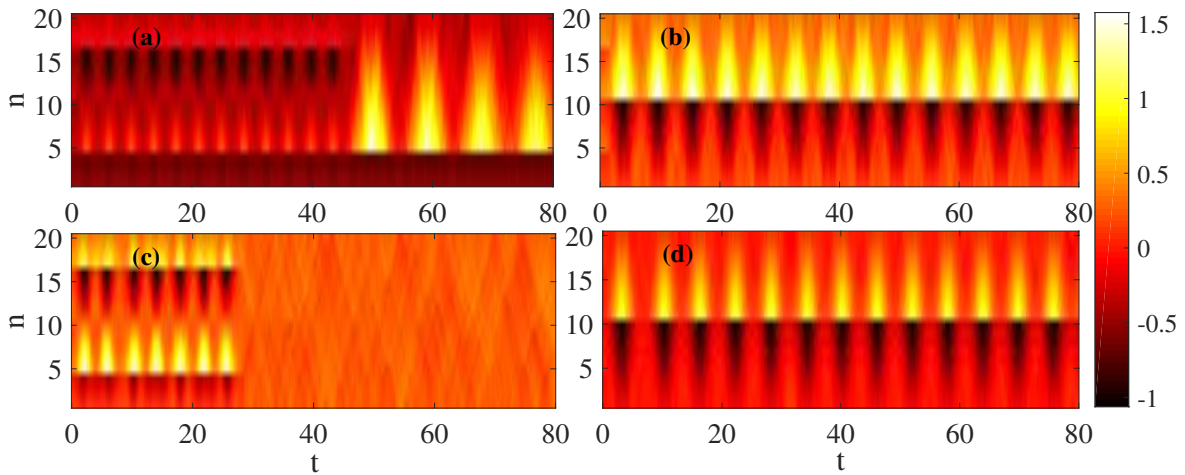}
		\caption{Same as Fig.~\ref{fig:dyn3s}, but
                  with perturbation
                  $\pm 0.05\times\hat{\mathbf{e}}_1$, for $d=1.07$. (a) Upper segment of the branch mode
                  $u_{3u} + 0.05\times\hat{\mathbf{e}}_1$; (b)  upper segment of the branch mode 
                  $u_{3u}- 0.05\times\hat{\mathbf{e}}_1$; (c) lower segment of the branch mode
                  $u_{3l}+ 0.05\times\hat{\mathbf{e}}_1$; (d) lower segment of the branch mode 
                  $u_{3l} - 0.05\times\hat{\mathbf{e}}_1$. The resulting
                  dynamics is more diverse, potentially leading to a
                  homogeneous state in (c), the survival of a central break
                  in (b) and (d), as well as the survival of one of the lateral
                breaks in (a).}
		\label{fig:dyn3s1}
	\end{center}
\end{figure}
\begin{figure}[H]
	\begin{center}
		\includegraphics[width=0.95\textwidth]{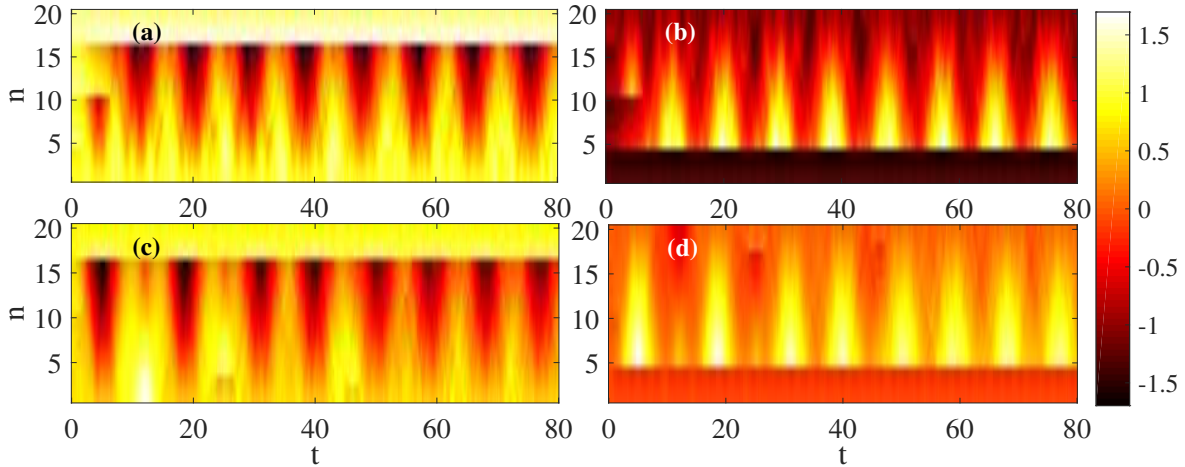}
		\caption{In this case, the 3 break waveform is
                  perturbed by $\pm 0.05\times\hat{\mathbf{e}}_2$,
                  for $d=1.07$. (a) Upper segment of the branch mode $u_{3u} + 0.05\times\hat{\mathbf{e}}_2$; (b)  upper segment of the branch mode 
                  $u_{3u}- 0.05\times\hat{\mathbf{e}}_2$; (c) lower segment of the branch mode 
                  $u_{3l}+ 0.05\times\hat{\mathbf{e}}_2$; (d) lower segment of the branch mode 
                  $u_{3l} - 0.05\times\hat{\mathbf{e}}_2$. In all cases, one of the lateral
                  breaks asymptotically
                persists.}
		\label{fig:dyn3s2}
	\end{center}
\end{figure}
The scenario of the evolution of a perturbed 3 break is shown
in  Fig.~\ref{fig:dyn3s}-\ref{fig:dyn3s2}. 
Here, as explained above, there are 2 unstable eigenvalues
present for all elements of these branches of solutions.
Therefore even more so than the 2 break case, the effects of $\hat{\mathbf{e}}_3$ are harder to see,
as any numerical noise contamination introducing $\hat{\mathbf{e}}_1$ and/or $\hat{\mathbf{e}}_2$ will
have stronger consequences. That is the reason why we chose to increase the strength of the
perturbation here compared to the 1 and 2 break cases.
The eigenvector $\hat{\mathbf{e}}_3$ is anti-symmetric like $u_3$. For the lower segment of the branch,
as noted before 
it will increase or decrease all breaks but more so the central one (as $\hat{\mathbf{e}}_3$ is larger
there). So if added, the central break grows at the expense of the side ones to lead to a  stable 
1 break (left bottom panel of Fig.~\ref{fig:dyn3s}). If subtracted it will collapse all three breaks 
to the elastic mode, yet
the extra energy will eventually allow the creation of a 1 break; right bottom panel of
Fig.~\ref{fig:dyn3s}. Note that the decay happens very soon
($t \approx 2.5$), so it is hardly discernible in panel (d).

In the case of the upper segment of the branch its third eigenmode has a central ``break" rather
smaller than the side ones;
see the left upper panel of Fig.~\ref{fig:u_ev3s}.
Thus, when added or subtracted to the stationary state, its influence is mainly on the side breaks,
leading them to oscillate, given the negative sign of $\lambda^2$.
This behavior, however, can only be seen for very short times. As previously mentioned,
contamination with any of the lower eigenmodes, especially
so the first which has the same parity, will lead to decay, governed mostly by those lower 
eigenmodes. This is evident on the dynamical simulation in
the upper panels of Fig.~\ref{fig:dyn3s}.

Turning now to the influence of the stronger eigenmodes, notice that $\hat{\mathbf{e}}_1$ is IP
 with the side breaks but OOP with the central break.
Then, in general, its effect will be to lead to the survival of the two side breaks by adding it, 
or the middle one by subtracting it. As we have
seen before the 2 break is also unstable, so one of those two will later collapse as well (see e.g. 
the top left panel of Fig.~\ref{fig:dyn3s1}). Notice the similarity between panel (b) and the upper
panels of Fig.~\ref{fig:dyn3s},
pointing to the influence of 
$\hat{\mathbf{e}}_1$ in that case. 

The effect of $\hat{\mathbf{e}}_2$, on the other end, being an even mode is nearly the same whether 
we add or subtract it to the mode. From its shape,
we can infer that it will collapse one of the side breaks, while increasing the other, and at an 
initial stage not influence much of the central one. But of course the 2 break thus formed is also
unstable and one (the central one in this case) will soon disappear as well towards
a 1 break state. This is confirmed in Fig.~\ref{fig:dyn3s2}.

\section{Conclusions \& Future Work}
\label{sec:conclusion}
In the present work, we have examined solutions involving different
numbers of fractures/breaks in a chain featuring a Lennard-Jones
potential of interaction between the nodes and Dirichlet boundary
conditions at the edges. We saw that for each of the solutions
beyond the uniform, elastic one, there was a (more) stable and
a (more) unstable portion of the branch, separated by a critical
point where the monotonicity of the strain and/or the energy
as a function of the precompression stress changed.
At the same time, while the single break solutions could be
potentially stable, any state with $N>1$ break would feature
$N-1$ real eigenvalue pairs, being associated with respective
instabilities. By monitoring  the eigendirections of these instabilities,
we could connect them with the tendency to eliminate one or more
breaks from the chain and result to fewer break, more robust waveforms.
These conclusions of the stability analysis were subsequently
corroborated by means of direct numerical simulations featuring the unstable
evolution of controlled numerical experiments where the instability-inducing
eigenvectors were added to the unstable structures.

Naturally, a number of additional directions for future research
are emerging as a result of the present study. On the one hand,
in the one-dimensional setting, it is especially relevant to
explore the role of interactions beyond those of nearest neighbors.
Inducing next nearest neighbor interactions in competition with
nearest neighbor ones may be a topic that will modify the
stability of the presently considered states and will be of
interest to explore in light of zigzag~\cite{efremidis} and
related configurations. On the other hand, it would be
of particular interest to explore how configurations like
the ones considered herein behave in higher dimensional
settings. The latter offer the possibility of different types
of geometries (e.g. in 2d square, hexagonal, honeycomb etc.)
and thus may induce an interplay of geometry with the nonlinear
interactions that may introduce novel states. Such studies are
currently in progress and will be reported in future publications.

\section{Aknowledgments}
A.S.R. acknowledges financial support from FCT through grant UID/FIS/04650/2013, and 
grant SFRH/BSAB/135213/2017.
PGK gratefully acknowledges
support from the US-AFOSR under FA9550-17-1-0114, as well as from
NSF-PHY-1602994.


\end{document}